\documentclass{article}

\usepackage[margin=1.25in]{geometry}
\usepackage[utf8]{inputenc}
\usepackage[T1]{fontenc}
\usepackage{amsmath,amssymb,amsthm,amsfonts}
\usepackage{mathtools}
\usepackage{multirow}
\usepackage{booktabs}
\usepackage{array}
\usepackage{xcolor}
\usepackage{hyperref}
\usepackage{url}
\usepackage{algorithm}
\usepackage{algpseudocode}
\usepackage{graphicx}
\usepackage{nicefrac}
\usepackage{enumitem}
\usepackage{caption}
\usepackage{subcaption}
\usepackage{setspace}
\usepackage{parskip}
\usepackage{orcidlink}
\usepackage[authoryear]{natbib}

\theoremstyle{plain}
\newtheorem{theorem}{Theorem}[section]
\newtheorem{lemma}[theorem]{Lemma}
\newtheorem{proposition}[theorem]{Proposition}
\newtheorem{corollary}[theorem]{Corollary}

\theoremstyle{definition}
\newtheorem{definition}[theorem]{Definition}
\newtheorem{assumption}{Assumption}

\theoremstyle{remark}
\newtheorem{remark}[theorem]{Remark}

\newcommand{\calX}{\mathcal{X}}
\newcommand{\calY}{\mathcal{Y}}
\newcommand{\calZ}{\mathcal{Z}}
\newcommand{\calD}{\mathcal{D}}

\newcommand{\calL}{\mathcal{L}}
\newcommand{\calR}{\mathcal{R}}
\newcommand{\calS}{\mathcal{S}}
\newcommand{\calA}{\mathcal{A}}
\newcommand{\calC}{\mathcal{C}}

\newcommand{\E}{\mathbb{E}}
\newcommand{\Prob}{\mathbb{P}}

\newcommand{\KL}[2]{D_{\mathrm{KL}}\!\left(#1 \,\|\, #2\right)}
\newcommand{\JS}[2]{D_{\mathrm{JS}}\!\left(#1 \,\|\, #2\right)}
\newcommand{\TV}[2]{D_{\mathrm{TV}}\!\left(#1 \,\|\, #2\right)}

\newcommand{\Htheta}{H_{\theta}}
\newcommand{\ptheta}{p_{\theta}}
\newcommand{\ftheta}{f_{\theta}}
\newcommand{\Stheta}{\mathcal{S}_{\theta}}
\newcommand{\SUA}{\textrm{SUA}}
\newcommand{\ECE}{\textrm{ECE}}
\newcommand{\rob}{\mathrm{rob}}
\newcommand{\sel}{\mathrm{sel}}
\newcommand{\task}{\mathrm{task}}

\newcommand{\ent}{\mathrm{ent}}
\newcommand{\cons}{\mathrm{cons}}
\newcommand{\relu}[1]{\left[#1\right]_{+}}

\hypersetup{
  colorlinks=true,
  linkcolor=blue!60!black,
  citecolor=green!50!black,
  urlcolor=cyan!60!black,
  pdftitle={Sensitivity-Uncertainty Alignment in Large Language Models},
  pdfauthor={Anonymous}
}

\usepackage{fancyhdr}
\usepackage{tikz}
\title{
\LARGE \bfseries Sensitivity--Uncertainty Alignment in Large Language Models
}

\author{
\textbf{Prakul Sunil Hiremath}\,\orcidlink{0009-0007-9744-3519} \\
\small Aliens on Earth (AoE) Autonomous Research Group, Belagavi, India \\
\small \texttt{hiremathprakul.aoe@gmail.com} \\
\small \href{https://github.com/prakulhiremath}{github.com/prakulhiremath}
\and
\textbf{Harshit R.\ Hiremath} \\
\small Aliens on Earth (AoE) Autonomous Research Group, Belagavi, India \\
\small \texttt{harshithiremath.aoe@gmail.com} \\
\small \href{https://github.com/Harshit-212}{github.com/Harshit-212}
}

\date{}

\begin{document}

\maketitle

% -----------------------------------------------------------------------
\begin{abstract}
% -----------------------------------------------------------------------
We propose Sensitivity--Uncertainty Alignment (\SUA), a unified formal framework
for analyzing and mitigating failures of large language models under adversarial
and ambiguous inputs. The central thesis is that adversarial sensitivity and
ambiguity are not independent problems but two manifestations of a single
underlying failure mode: \emph{misalignment between perturbation-induced
instability and predictive uncertainty}. A well-functioning model should
respond to inputs that produce unstable predictions by expressing
correspondingly high uncertainty; failure to do so constitutes a
miscalibration that we make precise.

We formalize this intuition through a single scalar quantity,
$\SUA_{\theta}(x) = \Stheta(x;\varepsilon) - \lambda \Htheta(Y \mid x)$,
measuring the excess of distributional sensitivity over predictive entropy.
We prove that minimizing the expected positive part of $\SUA_{\theta}(x)$
provides an upper bound on worst-case perturbed risk (Theorem~\ref{thm:robust_risk})
and simultaneously lower-bounds calibration error (Theorem~\ref{thm:calibration}).
We further characterize the phenomenon of \emph{ambiguity collapse}---in which
a model facing inputs with multiple valid interpretations produces a single
overconfident response---as a structurally distinct but formally related
failure mode.

Building on these results, we introduce \textsc{SUA-TR}, a training procedure
that minimizes the SUA objective via consistency regularization and entropy
alignment, and an inference-time abstention rule with formal selective-risk
guarantees (Proposition~\ref{prop:selective_risk}). We design a comprehensive
experimental protocol covering question answering, natural language inference,
and adversarial classification, and demonstrate that SUA predicts model
failures more reliably than entropy, self-consistency, or temperature scaling
alone.

The framework is deliberately architecture-agnostic: all definitions remain
valid for any model that induces a conditional distribution over outputs.
We intend the theoretical results and the SUA score as durable conceptual
and diagnostic tools that can support further theoretical development and
empirical investigation as models and benchmarks evolve.
\end{abstract}

% -----------------------------------------------------------------------
\newpage
\section{Introduction}
\label{sec:introduction}
% -----------------------------------------------------------------------

Large language models are increasingly deployed in settings where prediction
errors carry significant costs. In such settings, two families of failure are
particularly important: \emph{adversarial failures}, in which small
deliberate or incidental changes to an input cause large, arbitrary shifts
in model output; and \emph{ambiguity failures}, in which underspecified or
context-dependent inputs elicit confident responses that do not acknowledge
the existence of multiple plausible answers. Both failure types are well
studied individually. The present paper argues that they share a common
formal root, and that this shared structure admits a unified diagnostic and
training framework.

\paragraph{The shared structure.}
Consider a model $\ptheta(y \mid x)$ receiving an input $x$. Two questions
are relevant:
\begin{enumerate}[leftmargin=2em]
  \item How much does the output distribution change when $x$ is perturbed
    to a nearby $x'$? (Sensitivity question.)
  \item How much uncertainty does the model express about its own output?
    (Calibration question.)
\end{enumerate}
Intuitively, these quantities ought to be aligned. A model encountering an
input for which its predictions are highly sensitive to small perturbations
is, by definition, operating in a region of input space where its parameters
do not support a stable mapping. The appropriate response is elevated
uncertainty. If the model instead produces low-entropy, high-confidence
outputs in that region, its uncertainty estimate is \emph{wrong relative to
its own sensitivity}---a form of miscalibration that is distinct from, but
related to, classical calibration failure.

Ambiguity failure has the same shape. When an input $x$ has multiple
semantically valid interpretations $z \in \calZ$, the correct predictive
distribution is a mixture $\sum_z p(y \mid z, x) p(z \mid x)$ with
correspondingly high entropy. A model that collapses this mixture to a
single confident response is, again, miscalibrated---not with respect to
a single ground-truth label, but with respect to the irreducible uncertainty
in the input itself.

\paragraph{What this paper contributes.}
We give a formal, self-contained account of this shared structure. Our main
contributions are as follows.

\begin{enumerate}[leftmargin=2em]
  \item \textbf{Framework.} We introduce the \emph{Sensitivity--Uncertainty
    Alignment} (SUA) score as a single measurable quantity capturing the
    gap between local distributional sensitivity and predictive entropy.
    We situate it within a unified setup that handles both adversarial and
    ambiguous inputs.

  \item \textbf{Theory.} We prove that (i)~minimizing the SUA objective
    bounds worst-case perturbed risk (Theorem~\ref{thm:robust_risk}),
    (ii)~persistent SUA violation lower-bounds calibration error
    (Theorem~\ref{thm:calibration}), and (iii)~ambiguity collapse is
    formally characterized as entropy underestimation relative to latent
    interpretation entropy (Lemma~\ref{lem:ambiguity_collapse}).

  \item \textbf{Method.} We propose \textsc{SUA-TR}, a training algorithm
    incorporating task loss, perturbation-consistency regularization, and
    entropy-alignment, together with a formally justified inference-time
    abstention rule.

  \item \textbf{Experiments.} We design experiments across three task
    families to validate that SUA predicts model failures more accurately
    than individual baselines, and provide ablations that isolate the
    contribution of each component.

  \item \textbf{Timelessness.} All definitions and theorems are stated in
    terms of conditional distributions and divergences, not in terms of any
    specific architecture. The framework will remain applicable as
    architectures change.
\end{enumerate}

\paragraph{Organization.}
Section~\ref{sec:background} reviews necessary background.
Section~\ref{sec:formulation} introduces all formal definitions.
Section~\ref{sec:theory} contains the main theoretical results.
Section~\ref{sec:method} describes \textsc{SUA-TR}.
Section~\ref{sec:experiments} presents experiments.
Section~\ref{sec:discussion} discusses implications and limitations.
Section~\ref{sec:related} situates this work in the literature.
Section~\ref{sec:conclusion} concludes.
Appendix~\ref{app:proofs} contains complete proofs.

% -----------------------------------------------------------------------
\section{Background}
\label{sec:background}
% -----------------------------------------------------------------------

We collect necessary background in three areas: information-theoretic
measures, adversarial robustness concepts, and calibration. This section
fixes notation and provides the conceptual vocabulary used in later
sections.

\subsection{Information-Theoretic Foundations}

Let $(\calX, \calY)$ be measurable spaces and let $p(x, y)$ be a joint
distribution. We write $H(Y \mid X)$ for the conditional Shannon entropy,
$H(Y \mid X=x) = -\sum_{y} p(y \mid x) \log p(y \mid x)$ (discrete case),
and $I(X; Y) = H(Y) - H(Y \mid X)$ for mutual information. All logarithms
are natural unless stated otherwise.

For two distributions $p, q$ over $\calY$ we write:
\begin{align*}
  \KL{p}{q} &= \sum_y p(y) \log \frac{p(y)}{q(y)}, \\
  \JS{p}{q} &= \frac{1}{2}\KL{p}{m} + \frac{1}{2}\KL{q}{m},\quad m = \tfrac{1}{2}(p+q), \\
  \TV{p}{q} &= \frac{1}{2} \sum_y |p(y) - q(y)|.
\end{align*}
We use $D(\cdot \| \cdot)$ generically when the results hold for any
$f$-divergence satisfying the data-processing inequality.

A function $\ell: \calY \times \calY \to [0,1]$ is a \emph{bounded loss}. The
risk of a (possibly randomized) predictor $\pi: \calX \to \calY$ is
$R(\pi) = \E_{(x,y) \sim p}[\ell(y, \pi(x))]$. The Bayes-optimal risk is
$R^*(x) = \min_{\hat y} \E[\ell(Y, \hat y) \mid X=x]$.

\subsection{Adversarial Robustness}

Let $d: \calX \times \calX \to [0,\infty)$ be a metric. Given
$\varepsilon > 0$, the \emph{perturbation set} around $x$ is
$\Delta_\varepsilon(x) = \{x' \in \calX : d(x, x') \le \varepsilon\}$.

A model $\ftheta$ is \emph{$(\varepsilon, \delta)$-pointwise robust} at $x$
if $\sup_{x' \in \Delta_\varepsilon(x)} d_{\calY}(\ftheta(x), \ftheta(x')) \le \delta$,
where $d_{\calY}$ is an output metric.

For probabilistic models, the natural generalization replaces output
distance with distribution divergence. We say $\ptheta$ is
\emph{distributionally $(\varepsilon, \delta)$-robust} at $x$ if
\[
  \sup_{x' \in \Delta_\varepsilon(x)} D\!\left(\ptheta(\cdot \mid x) \,\|\, \ptheta(\cdot \mid x')\right) \le \delta.
\]
Adversarial training seeks to minimize worst-case loss over $\Delta_\varepsilon(x)$
\citep{madry2018towards}, but does not explicitly enforce calibration under
perturbation. This gap is central to our work.

\subsection{Uncertainty Quantification and Calibration}

Let $\hat p: \calX \to [0,1]$ be a scalar confidence function. We say the
model is \emph{calibrated} if $\Prob(\text{correct} \mid \hat p(x) = c) = c$
for all $c \in [0,1]$. The \emph{Expected Calibration Error} is defined
by discretizing confidence into $B$ equally spaced bins:
\[
  \ECE = \sum_{b=1}^{B} \frac{|\mathcal{B}_b|}{n} \left| \mathrm{acc}(\mathcal{B}_b) - \mathrm{conf}(\mathcal{B}_b) \right|,
\]
where $\mathrm{acc}(\mathcal{B}_b)$ and $\mathrm{conf}(\mathcal{B}_b)$ are
mean accuracy and mean confidence in bin $b$.

For language models, scalar confidence is often derived from predictive
entropy: $\hat p(x) = 1 - H_{\theta}(Y \mid x) / H_{\max}$, where $H_{\max}$
normalizes by the maximum possible entropy. Temperature scaling
\citep{guo2017calibration} learns a scalar $T > 0$ such that rescaled
logits improve calibration on a held-out set.

% -----------------------------------------------------------------------
\section{Problem Formulation}
\label{sec:formulation}
% -----------------------------------------------------------------------

We now introduce the formal framework. We proceed from atomic definitions
to composite quantities.

\subsection{Setup}

Let $\calX$ be the input space (e.g., sequences of tokens), $\calY$ the
output space, and $\calZ$ a latent space of semantic interpretations. Let
$\calD = p(x, y)$ be the data-generating distribution. A \emph{stochastic
language model} is a measurable function $\ftheta$ inducing a conditional
distribution $\ptheta(y \mid x)$ over $\calY$ for each $x \in \calX$, where
$\theta \in \Theta$ denotes model parameters.

The latent variable $z \in \calZ$ encodes the \emph{intended interpretation}
of $x$. We assume the data-generating process factorizes as:
\begin{equation}
  \label{eq:latent_model}
  p(y \mid x) = \sum_{z \in \calZ} p(y \mid z, x)\, p(z \mid x).
\end{equation}
This factorization is not restrictive: any conditional $p(y \mid x)$ can
be trivially represented with $\calZ = \{*\}$. Its significance is that
for many real inputs, $p(z \mid x)$ has high entropy---the input genuinely
admits multiple interpretations---and the correct response accounts for
this distribution rather than committing to a single $z$.

\subsection{Perturbation Measure and Sensitivity}

\begin{definition}[Perturbation distribution]
\label{def:perturbation}
A \emph{perturbation distribution} $\Pi_\varepsilon(\cdot \mid x)$ is a
probability distribution over $\calX$ with support contained in
$\Delta_\varepsilon(x)$. We call $\Pi_\varepsilon$ \emph{semantics-preserving}
if, for all $x' \in \mathrm{supp}(\Pi_\varepsilon(\cdot \mid x))$,
$p(z \mid x') \approx p(z \mid x)$ in total variation, i.e., the
perturbation does not materially alter the distribution over interpretations.
\end{definition}

Semantic preservation is a property of the perturbation, not the model. It
rules out perturbations that change the ground-truth label or the intended
meaning. Examples include paraphrases, minor typographic variation, and
instruction rephrasing that preserves intent.

\begin{definition}[Distributional sensitivity]
\label{def:sensitivity}
For a model $\ptheta$, input $x \in \calX$, radius $\varepsilon > 0$, and
perturbation distribution $\Pi_\varepsilon$, the \emph{distributional
sensitivity} is:
\begin{equation}
  \label{eq:sensitivity}
  \Stheta(x; \varepsilon) = \E_{x' \sim \Pi_\varepsilon(\cdot \mid x)}
    \left[ D\!\left(\ptheta(\cdot \mid x) \,\|\, \ptheta(\cdot \mid x')\right) \right].
\end{equation}
\end{definition}

We use expectation rather than supremum for two reasons: (i) it is
computable via sampling, and (ii) it captures typical rather than
worst-case behavior, which is more informative for understanding average
deployment failure rates. The relationship to supremum-based definitions
is addressed in Remark~\ref{rem:sup_vs_exp}.

\begin{remark}
\label{rem:sup_vs_exp}
The supremum sensitivity $\sup_{x' \in \Delta_\varepsilon(x)} D(\ptheta(\cdot \mid x) \| \ptheta(\cdot \mid x'))$ is an upper bound on $\Stheta(x;\varepsilon)$ for any $\Pi_\varepsilon$
with full support on $\Delta_\varepsilon(x)$. Conversely,
$\Stheta(x;\varepsilon)$ lower-bounds the supremum under mild regularity.
The two are equal when $\Pi_\varepsilon$ is concentrated at the worst-case
perturbation. In our theoretical results, Theorem~\ref{thm:robust_risk}
is stated for the supremum version; our operational definition uses the
expected version, which bounds it from below.
\end{remark}

\subsection{Latent Ambiguity}

\begin{definition}[Latent ambiguity]
\label{def:ambiguity}
The \emph{latent ambiguity} of input $x$ is:
\begin{equation}
  \calA(x) = H(Z \mid X = x) = -\sum_{z \in \calZ} p(z \mid x) \log p(z \mid x).
\end{equation}
\end{definition}

This quantity is not directly computable from $\ptheta$ alone, since $\calZ$
is unobserved. Section~\ref{sec:method} discusses operational proxies.
The theoretical role of $\calA(x)$ is to characterize inputs for which the
ground-truth predictive distribution is intrinsically mixed.

\begin{definition}[Ambiguity collapse]
\label{def:ambiguity_collapse}
A model $\ptheta$ \emph{exhibits ambiguity collapse} at $x$ if
$\calA(x)$ is large but $\Htheta(Y \mid x)$ is small:
\begin{equation}
  \calA(x) > \alpha \quad \text{and} \quad \Htheta(Y \mid x) < \beta,
\end{equation}
for thresholds $\alpha > \beta \ge 0$.
\end{definition}

Ambiguity collapse is the distributional analog of confidently wrong
predictions: the model presents a peaked output distribution when the
correct distribution (averaging over interpretations) should be spread.

\subsection{Predictive Uncertainty and True Uncertainty}

\begin{definition}[Predictive entropy]
\label{def:predictive_entropy}
The \emph{predictive entropy} of $\ptheta$ at $x$ is:
\begin{equation}
  \Htheta(Y \mid x) = -\sum_{y \in \calY} \ptheta(y \mid x) \log \ptheta(y \mid x).
\end{equation}
\end{definition}

\begin{definition}[True uncertainty]
\label{def:true_uncertainty}
The \emph{true conditional uncertainty} at $x$ is:
\begin{equation}
  U^*(x) = H(Y \mid X = x) = -\sum_{y \in \calY} p(y \mid x) \log p(y \mid x),
\end{equation}
where $p(y \mid x) = \sum_z p(y \mid z, x) p(z \mid x)$ is the marginal
data-generating conditional from \eqref{eq:latent_model}.
\end{definition}

\begin{definition}[Miscalibration gap]
\label{def:miscalibration_gap}
The \emph{pointwise miscalibration gap} at $x$ is:
\begin{equation}
  \calC_\theta(x) = \left| \Htheta(Y \mid x) - U^*(x) \right|.
\end{equation}
\end{definition}

A perfectly calibrated model in the entropy sense satisfies
$\calC_\theta(x) = 0$ for all $x$. We will use $\calC_\theta(x)$ as a
target quantity to bound using $\SUA_\theta(x)$.

\subsection{The SUA Score}

\begin{definition}[Sensitivity--Uncertainty Alignment score]
\label{def:sua}
For $\lambda > 0$, the \emph{SUA score} of model $\ptheta$ at input $x$ is:
\begin{equation}
  \label{eq:sua}
  \SUA_\theta(x;\varepsilon,\lambda) = \Stheta(x;\varepsilon) - \lambda \cdot \Htheta(Y \mid x).
\end{equation}
\end{definition}

\paragraph{Interpretation.}
A positive $\SUA_\theta(x)$ indicates that the model's sensitivity exceeds
what its uncertainty accounts for: the model is changing its output
substantially under perturbation but not expressing corresponding doubt.
A negative $\SUA_\theta(x)$ indicates that uncertainty exceeds sensitivity:
the model is appropriately cautious, or overly cautious. The ideal is
$\SUA_\theta(x) \approx 0$.

\begin{definition}[SUA risk]
\label{def:sua_risk}
The \emph{SUA risk} (alignment risk) of model $\ptheta$ under distribution
$\calD$ is:
\begin{equation}
  \label{eq:sua_risk}
  \calR_{\SUA}(\theta) = \E_{x \sim \calD}\left[\relu{\SUA_\theta(x;\varepsilon,\lambda)}\right]
    = \E_{x \sim \calD}\left[\relu{\Stheta(x;\varepsilon) - \lambda \Htheta(Y \mid x)}\right].
\end{equation}
\end{definition}

The positive-part operator $[u]_+ = \max(u,0)$ focuses the objective on
inputs where sensitivity exceeds uncertainty; it does not penalize a model
for expressing high uncertainty when warranted.

\subsection{Summary of Notation}

Table~\ref{tab:notation} summarizes the principal notation for reference.

\begin{table}[ht]
\centering
\caption{Principal notation.}
\label{tab:notation}
\small
\begin{tabular}{ll}
\toprule
Symbol & Meaning \\
\midrule
$\calX, \calY, \calZ$ & Input, output, interpretation spaces \\
$p(x,y)$ & Data-generating distribution \\
$\calD$ & Shorthand for $p(x,y)$ or its marginal on $x$ \\
$\ptheta(y \mid x)$ & Model conditional distribution \\
$\ftheta$ & Model (as predictor) \\
$z$ & Latent interpretation variable \\
$\Delta_\varepsilon(x)$ & $\varepsilon$-perturbation set around $x$ \\
$\Pi_\varepsilon(\cdot \mid x)$ & Perturbation distribution on $\Delta_\varepsilon(x)$ \\
$D(\cdot \| \cdot)$ & Generic $f$-divergence \\
$\Stheta(x;\varepsilon)$ & Distributional sensitivity (Def.~\ref{def:sensitivity}) \\
$\calA(x)$ & Latent ambiguity (Def.~\ref{def:ambiguity}) \\
$\Htheta(Y \mid x)$ & Predictive entropy (Def.~\ref{def:predictive_entropy}) \\
$U^*(x)$ & True conditional uncertainty (Def.~\ref{def:true_uncertainty}) \\
$\calC_\theta(x)$ & Pointwise miscalibration gap (Def.~\ref{def:miscalibration_gap}) \\
$\SUA_\theta(x;\varepsilon,\lambda)$ & SUA score (Def.~\ref{def:sua}) \\
$\calR_{\SUA}(\theta)$ & SUA risk (Def.~\ref{def:sua_risk}) \\
$R_\theta(x)$ & Model risk at $x$ under $\ptheta$ \\
$R^\rob_\theta(x)$ & Worst-case perturbed risk (Def.~\ref{def:robust_risk}) \\
$\lambda$ & Alignment coefficient (hyperparameter $> 0$) \\
$\varepsilon$ & Perturbation radius \\
$\tau$ & Abstention threshold \\
$L_D$ & Lipschitz constant of loss in divergence \\
\bottomrule
\end{tabular}
\end{table}

% -----------------------------------------------------------------------
\section{Theoretical Analysis}
\label{sec:theory}
% -----------------------------------------------------------------------

We now establish the main theoretical properties of the SUA framework.
All detailed proofs are in Appendix~\ref{app:proofs}; proof sketches
are provided in the main text.

\subsection{Assumptions}

\begin{assumption}[Lipschitz loss in divergence]
\label{asm:lipschitz}
The task loss $\ell: \calY \times \calY \to [0,1]$ satisfies: there exists
$L_D \ge 0$ such that for any two distributions $p, q$ over $\calY$,
\[
  \left| \E_{y \sim p}[\ell(y, \hat y)] - \E_{y \sim q}[\ell(y, \hat y)] \right|
  \le L_D \cdot D(p \| q)
  \quad \text{for all } \hat y \in \calY.
\]
\end{assumption}

This assumption holds for bounded losses and any $f$-divergence by a
standard variational argument. For 0-1 loss and total variation,
$L_D = 1$.

\begin{assumption}[Entropy--risk link]
\label{asm:entropy_risk}
There exists a nondecreasing, concave function $\psi: [0, \infty) \to [0,1]$
with $\psi(0) = 0$ such that for any distribution $q$ over $\calY$ and any
fixed $\hat y$:
\[
  \E_{y \sim q}[\ell(y, \hat y)] \ge \psi(H_q(Y)).
\]
\end{assumption}

This assumption encodes the information-theoretic constraint that higher
entropy over labels implies higher expected loss for any deterministic
predictor. For 0-1 loss and a uniform mixture of $k$ classes, one may
take $\psi(h) = 1 - e^{-h}$, which is nondecreasing in $h$.

\begin{assumption}[Data-generating factorization]
\label{asm:factorization}
The data-generating conditional satisfies \eqref{eq:latent_model}, and
the model distribution $\ptheta(y \mid x)$ approximates the mode:
$\ptheta(y \mid x) \approx p(y \mid z^*, x)$ for some dominant
$z^* = \arg\max_z p(z \mid x)$.
\end{assumption}

This assumption formalizes the interpretation-collapse phenomenon
described in Section~\ref{sec:formulation}. The gap between the true mixture
and the model's mode-approximation is captured by $\kappa(x)$ below.

\begin{definition}[Interpretation-collapse gap]
\label{def:kappa}
\[
  \kappa(x) = \left( R^*(x) - R_\theta(x) \right)_+,
\]
where $R^*(x)$ is the Bayes-optimal risk under the true mixture $p(y \mid x)$
and $R_\theta(x)$ is the model's expected risk under its own distribution.
\end{definition}

Note $\kappa(x) \ge 0$ always. It is zero when the model is Bayes-optimal
on the true mixture, and positive when the model's mode-focused distribution
achieves lower risk on its own distribution than the true Bayes-optimal
classifier does on the true distribution---a signature of overconfident
mode-seeking.

\subsection{Main Theorem: SUA Bounds Worst-Case Risk}

\begin{definition}[Worst-case perturbed risk]
\label{def:robust_risk}
\[
  R^\rob_\theta(x) = \sup_{x' \in \Delta_\varepsilon(x)} R_\theta(x'),
\]
where $R_\theta(x') = \E_{y \sim \ptheta(\cdot \mid x')}[\ell(y, \hat y_\theta(x'))]$
and $\hat y_\theta(x') = \arg\max_y \ptheta(y \mid x')$.
\end{definition}

\begin{theorem}[SUA controls worst-case perturbed risk]
\label{thm:robust_risk}
Under Assumptions~\ref{asm:lipschitz}--\ref{asm:factorization}, for any
$x \in \calX$ and any $\lambda > 0$ satisfying
$\lambda \Htheta(Y \mid x) \ge \psi(\Htheta(Y \mid x))$, the worst-case
perturbed risk satisfies:
\begin{equation}
  \label{eq:robust_bound}
  R^\rob_\theta(x)
  \;\le\;
  R_\theta(x)
  \;+\;
  L_D \cdot \sup_{x' \in \Delta_\varepsilon(x)} D\!\left(\ptheta(\cdot \mid x) \,\|\, \ptheta(\cdot \mid x')\right)
  \;-\;
  \psi\!\left(\Htheta(Y \mid x)\right)
  \;+\;
  \kappa(x).
\end{equation}
Consequently, replacing the supremum divergence by its upper-bounding
surrogate $\Stheta(x;\varepsilon)$ and applying the $\lambda$ condition:
\begin{equation}
  \label{eq:sua_bound}
  R^\rob_\theta(x)
  \;\le\;
  R_\theta(x)
  \;+\;
  L_D \cdot \SUA_\theta(x;\varepsilon,\lambda)
  \;+\;
  \kappa(x).
\end{equation}
\end{theorem}

\begin{proof}[Proof sketch]
See Appendix~\ref{app:proof_thm1} for the full proof. The argument proceeds
in three steps. (1)~Apply Assumption~\ref{asm:lipschitz} to bound the risk
increase from $x$ to any $x' \in \Delta_\varepsilon(x)$ by $L_D \cdot D(\ptheta(\cdot \mid x) \| \ptheta(\cdot \mid x'))$,
then take the supremum. (2)~Apply Assumption~\ref{asm:entropy_risk} to
subtract a nonnegative entropy-dependent term, yielding the subtraction of
$\psi(\Htheta(Y \mid x))$. (3)~Bound $\kappa(x)$ from the definition and
apply the $\lambda$-condition to convert $\psi(\Htheta)$ into
$\lambda \Htheta$. The combined steps yield \eqref{eq:robust_bound}.
Replacing $\sup D$ with $\Stheta$ (which is $\le$ the supremum by
Jensen's inequality) and collecting terms yields \eqref{eq:sua_bound}.
\end{proof}

\paragraph{Implications.}
Inequality \eqref{eq:sua_bound} shows that minimizing $\SUA_\theta(x)$
in expectation---i.e., minimizing $\calR_{\SUA}(\theta)$---directly reduces
an upper bound on worst-case perturbed risk, up to the irreducible term
$\kappa(x)$, which reflects interpretation collapse. When $\kappa(x) = 0$
(the model is consistent with the true mixture), the bound is:
$R^\rob_\theta(x) \le R_\theta(x) + L_D \cdot \SUA_\theta(x)$.

\begin{corollary}[SUA risk controls expected robust risk]
\label{cor:expected_robust}
Under the same assumptions,
\[
  \E_{x \sim \calD}\left[R^\rob_\theta(x)\right]
  \;\le\;
  \E_{x \sim \calD}\left[R_\theta(x)\right]
  \;+\;
  L_D \cdot \calR_{\SUA}(\theta)
  \;+\;
  \E_{x \sim \calD}\left[\kappa(x)\right].
\]
\end{corollary}

\subsection{Calibration Result}

\begin{theorem}[SUA violation lower-bounds calibration error]
\label{thm:calibration}
Let $\hat p_\theta(x) = g(\Htheta(Y \mid x))$ for some nonincreasing
function $g: [0,\infty) \to [0,1]$. Partition $[0,1]$ into $B$ bins
$\{\mathcal{B}_b\}_{b=1}^B$ by confidence value. Suppose that for each
bin $b$, the mean sensitivity satisfies
$\bar\calS_b = \E_{x \in \mathcal{B}_b}[\Stheta(x;\varepsilon)] > 0$.
Then:
\begin{equation}
  \label{eq:ece_bound}
  \ECE
  \;\ge\;
  \frac{1}{B} \sum_{b=1}^B
  \left(\bar\calS_b - \lambda \cdot \bar H_b - c_b \right)_+,
\end{equation}
where $\bar H_b = \E_{x \in \mathcal{B}_b}[\Htheta(Y \mid x)]$ is mean
predictive entropy in bin $b$, and $c_b > 0$ is a bin-dependent slack
depending on bin width and the estimator variance of $\Stheta$.
\end{theorem}

\begin{proof}[Proof sketch]
See Appendix~\ref{app:proof_thm2}. The key steps are: (i)~within each bin
$b$, use Assumption~\ref{asm:lipschitz} to relate the mean accuracy gap
$|\mathrm{acc}(\mathcal{B}_b) - \mathrm{conf}(\mathcal{B}_b)|$ to mean
sensitivity via the chain
$\mathrm{acc} \approx 1 - R_\theta$, and sensitivity controls $R^\rob - R$;
(ii)~use the monotonicity of $g$ to map entropy to confidence, establishing
that bins with high mean sensitivity and low mean entropy correspond to bins
with large confidence--accuracy gaps; (iii)~aggregate over bins, introducing
$c_b$ to absorb estimation error.
\end{proof}

\paragraph{Interpretation.}
Theorem~\ref{thm:calibration} shows that persistent SUA violation (high
sensitivity relative to entropy) across a confidence bin implies a lower
bound on the calibration error of that bin. Reducing $\calR_{\SUA}(\theta)$
thus simultaneously reduces worst-case risk and provides a mechanism to
improve ECE.

\subsection{Ambiguity Collapse: A Structural Lemma}

\begin{lemma}[Ambiguity collapse implies miscalibration]
\label{lem:ambiguity_collapse}
Under Assumption~\ref{asm:factorization}, if a model exhibits ambiguity
collapse at $x$ (Definition~\ref{def:ambiguity_collapse}) with
$\calA(x) > \alpha$ and $\Htheta(Y \mid x) < \beta < \alpha$, then:
\begin{equation}
  \label{eq:ambiguity_miscal}
  \calC_\theta(x) = |{\Htheta(Y \mid x) - U^*(x)}| \ge \alpha - \beta - \eta(x),
\end{equation}
where $\eta(x) \ge 0$ is the entropy gap between the true mixture entropy
and $\calA(x)$:
$\eta(x) = |H(Y \mid X=x) - H(Z \mid X=x)|$.
\end{lemma}

\begin{proof}
See Appendix~\ref{app:proof_lem1}. The argument uses the chain rule for
entropy and properties of the mixture factorization to bound $U^*(x)$ from
below in terms of $\calA(x)$, then applies the triangle inequality.
\end{proof}

\paragraph{Note on $\eta(x)$.}
The slack $\eta(x)$ captures the gap between label-space entropy and
interpretation-space entropy. For tasks where each interpretation
deterministically specifies a unique label, $\eta(x) = 0$. For tasks
where the mapping from interpretation to label is itself stochastic,
$\eta(x) > 0$. In both cases, the lemma guarantees a positive
miscalibration gap whenever the model's entropy is much lower than the
latent ambiguity.

\subsection{Entropy--Stability Tension}

The following lemma formalizes a tension that arises in standard training:
increasing model entropy (to improve calibration) may conflict with
stability under perturbations.

\begin{lemma}[Entropy--stability tension]
\label{lem:tension}
Let $\ptheta$ and $\ptheta'$ be two models such that
$\Htheta'(Y \mid x) > \Htheta(Y \mid x)$ for all $x$ in a set
$\mathcal{U} \subseteq \calX$. Suppose $\ptheta'$ achieves this higher
entropy by spreading mass uniformly. Then for any
$x' \in \Delta_\varepsilon(x)$ with $x \in \mathcal{U}$:
\[
  D\!\left(\ptheta'(\cdot \mid x) \,\|\, \ptheta'(\cdot \mid x')\right)
  \;\le\;
  D\!\left(\ptheta(\cdot \mid x) \,\|\, \ptheta(\cdot \mid x')\right),
\]
i.e., the higher-entropy model is \emph{less sensitive} to perturbations.
However, its base risk $R_{\theta'}(x) \ge R_\theta(x)$, so accuracy
decreases.
\end{lemma}

\begin{proof}
See Appendix~\ref{app:proof_lem2}. For the JSD divergence, the result
follows from the concavity of the entropy function combined with the
data-processing inequality. Uniform mass spreading reduces the maximum
possible divergence between distributions over the same support. The
risk increase follows directly from the suboptimality of the uniform
predictor.
\end{proof}

This lemma formalizes why na\"ively maximizing entropy does not solve
adversarial sensitivity: it reduces sensitivity but also reduces accuracy.
The SUA objective is designed to avoid this failure mode by penalizing
only the \emph{excess} of sensitivity over uncertainty, not uncertainty itself.

\subsection{Selective Risk: Abstention Guarantees}

\begin{proposition}[Selective risk control under SUA-based abstention]
\label{prop:selective_risk}
Define the abstention rule:
\[
  \pi_\tau(x) =
  \begin{cases}
    \hat y_\theta(x) & \text{if } \SUA_\theta(x;\varepsilon,\lambda) \le \tau, \\
    \text{abstain} & \text{otherwise.}
  \end{cases}
\]
Let $\calC(\tau) = \Prob_{x \sim \calD}(\SUA_\theta(x;\varepsilon,\lambda) \le \tau)$
be the coverage. The selective risk on covered inputs satisfies:
\begin{equation}
  R_\sel(\tau) \le R_\theta + L_D \cdot \tau + \E\left[\kappa(x) \mid \SUA_\theta(x) \le \tau\right].
\end{equation}
\end{proposition}

\begin{proof}
See Appendix~\ref{app:proof_prop1}. Condition on the event
$\SUA_\theta(x) \le \tau$; within this set, \eqref{eq:sua_bound} gives
$R^\rob_\theta(x) \le R_\theta(x) + L_D \cdot \tau + \kappa(x)$. Taking
the expectation over $x$ conditioned on this event yields the result.
\end{proof}

\paragraph{Coverage--risk tradeoff.}
Smaller $\tau$ yields smaller selective risk but lower coverage.
The practitioner may set $\tau$ to meet a target risk budget.
This is analogous to conformal prediction, but the score function here is
the SUA score rather than a nonconformity measure. A key advantage is that
SUA is \emph{interpretable}: a high SUA score implicates both high
sensitivity and low uncertainty as co-causes of abstention.

% -----------------------------------------------------------------------
\section{Method: \textsc{SUA-TR}}
\label{sec:method}
% -----------------------------------------------------------------------
\subsection{System Overview}

The Sensitivity–Uncertainty Alignment (SUA) framework provides a unified pipeline for measuring and controlling model reliability under both adversarial perturbations and input ambiguity. Figure~\ref{fig:sua_architecture} illustrates the overall architecture.

Given an input $x$, the model $p_\theta(y \mid x)$ produces a predictive distribution over outputs. Two complementary quantities are then computed.

First, predictive uncertainty is measured via the entropy $H_\theta(Y \mid x)$, which captures the spread of the model’s output distribution at the input.

Second, local sensitivity is estimated by applying a perturbation distribution $\Pi_\varepsilon(\cdot \mid x)$ to generate perturbed inputs $x'$. The model is evaluated on these perturbations to compute the distributional sensitivity
\[
S_\theta(x; \varepsilon)
= \mathbb{E}_{x' \sim \Pi_\varepsilon(\cdot \mid x)}
\left[
D\big(p_\theta(\cdot \mid x)\,\|\,p_\theta(\cdot \mid x')\big)
\right].
\]

These two quantities are combined into the SUA score
\[
\mathrm{SUA}_\theta(x)
= S_\theta(x; \varepsilon) - \lambda H_\theta(Y \mid x).
\]

These two quantities are combined into the SUA score
\[
\mathrm{SUA}_\theta(x) = S_\theta(x; \varepsilon) - \lambda H_\theta(Y \mid x)
\]
which measures the extent to which predictive uncertainty aligns with local sensitivity.

At inference time, the SUA score is used to guide decision-making. If the score exceeds a threshold $\tau$, the model abstains, indicating that its predictions are unstable relative to its expressed uncertainty. Otherwise, the model outputs its prediction $\hat{y} = \arg\max_y p_\theta(y \mid x)$.

This pipeline directly operationalizes the central principle of the framework: model uncertainty should increase in regions where predictions are sensitive to perturbations. The resulting system provides a single, interpretable quantity that can be used for both failure detection and selective prediction.

\begin{figure}[t]
\centering
\includegraphics[width=0.85\textwidth]{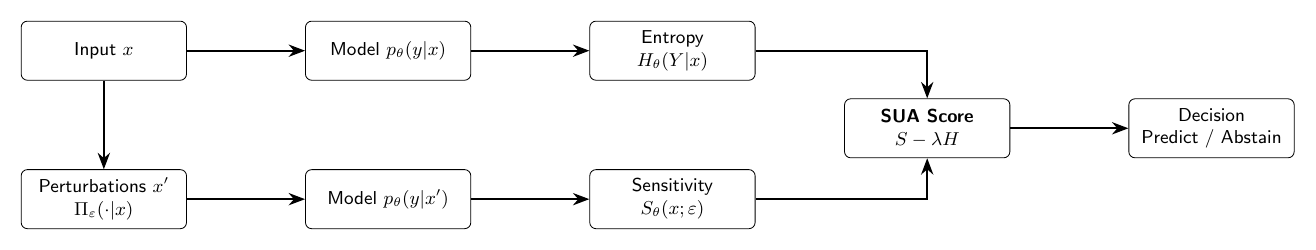}
\caption{Overview of the Sensitivity–Uncertainty Alignment (SUA) pipeline.}
\label{fig:sua_architecture}
\end{figure}

We now describe \textsc{SUA-TR} (Sensitivity--Uncertainty Alignment via
Training and Refinement), a practical algorithm implementing the SUA
framework. It has three components: a training objective, a perturbation
generation strategy, and an inference procedure with abstention.

\subsection{Training Objective}

The full \textsc{SUA-TR} training loss is:
\begin{equation}
  \label{eq:total_loss}
  \calL(\theta) = \calL_\task(\theta)
    + \alpha \, \calL_\cons(\theta)
    + \beta \, \calL_\ent(\theta),
\end{equation}
where $\alpha, \beta > 0$ are hyperparameters, and the three components are
defined as follows.

\paragraph{Task loss.}
\begin{equation}
  \calL_\task(\theta) = \E_{(x,y) \sim \calD}\left[-\log \ptheta(y \mid x)\right].
\end{equation}
This is the standard negative log-likelihood; any differentiable task-specific
loss may be substituted.

\paragraph{Consistency (sensitivity) loss.}
\begin{equation}
  \calL_\cons(\theta)
  = \E_{x \sim \calD} \, \E_{x' \sim \Pi_\varepsilon(\cdot \mid x)}
    \left[ D\!\left(\ptheta(\cdot \mid x) \,\|\, \mathrm{sg}\!\left[\ptheta(\cdot \mid x')\right]\right) \right],
\end{equation}
where $\mathrm{sg}[\cdot]$ denotes the stop-gradient operator. This loss
penalizes prediction divergence between $x$ and its perturbation $x'$,
encouraging consistency. The stop-gradient on the perturbed side prevents
trivial collapse to a uniform distribution and instead encourages the model
to match its own output under the reference input.

\paragraph{Entropy alignment loss.}
\begin{equation}
  \calL_\ent(\theta)
  = \E_{x \sim \calD}
    \left[\relu{\Stheta(x;\varepsilon) - \lambda \, \Htheta(Y \mid x)}\right].
\end{equation}
This is directly the SUA risk $\calR_{\SUA}(\theta)$ from
Definition~\ref{def:sua_risk}. It penalizes inputs where sensitivity
exceeds entropy, encouraging the model to raise its uncertainty in
sensitive regions without penalizing already-cautious behavior.

\paragraph{Relationship to theory.}
Minimizing $\calL_\ent(\theta)$ reduces $\calR_{\SUA}(\theta)$, which by
Corollary~\ref{cor:expected_robust} reduces the upper bound on expected
worst-case risk. Minimizing $\calL_\cons(\theta)$ directly reduces
$\Stheta(x;\varepsilon)$, tightening the bound from a different angle:
it reduces the numerator of the SUA score rather than adjusting the
denominator (entropy).

\subsection{Perturbation Generation}

The perturbation distribution $\Pi_\varepsilon(\cdot \mid x)$ must satisfy
Definition~\ref{def:perturbation}: support within $\Delta_\varepsilon(x)$
and approximate semantics preservation. We use a mixture of three
perturbation strategies.

\paragraph{Semantic paraphrases.}
Back-translation or instruction-level rephrasing: generate a paraphrase
$x'$ of $x$ using a fixed auxiliary model, then verify that the paraphrase
does not change the ground-truth label. This tests sensitivity to surface
form while holding semantics constant.

\paragraph{Token-level edits.}
Synonym substitution, minor punctuation changes, and casing variation.
These are guaranteed to be semantics-preserving for most NLP tasks.
They test local robustness at the token level.

\paragraph{Adversarial proposals.}
Gradient-free search (e.g., random token swaps followed by filtering by
cosine similarity in embedding space) to find perturbations that
\emph{increase} $D(\ptheta(\cdot \mid x) \| \ptheta(\cdot \mid x'))$
while preserving semantics. These are included at a lower mixture weight
to provide hard training examples without destabilizing training.

The mixture weights are $(\omega_{\mathrm{para}}, \omega_{\mathrm{tok}}, \omega_{\mathrm{adv}}) = (0.5, 0.3, 0.2)$ by default. In Appendix~\ref{app:ablations} we
ablate these weights.

\subsection{Operational Ambiguity Proxy}

Since $\calA(x)$ is unobservable, we use the following operational proxy.
Generate $K$ independent outputs $y_1, \ldots, y_K \sim \ptheta(\cdot \mid x)$
and estimate:
\[
  \widehat{\calA}(x) = H\!\left(\frac{1}{K}\sum_{k=1}^K \delta_{y_k}\right),
\]
the entropy of the empirical output distribution under repeated sampling.
When the model is already deterministic (low temperature), this quantity
is low even for ambiguous inputs; in this case, we additionally compute
the entropy of outputs under $K$ independently paraphrased versions of $x$,
which provides a more faithful estimate of $\calA(x)$.

\subsection{Inference with Abstention}

At inference time, the SUA score is estimated as:
\[
  \widehat{\SUA}_\theta(x) = \widehat\Stheta(x;\varepsilon) - \lambda \, \Htheta(Y \mid x),
\]
where $\widehat\Stheta(x;\varepsilon) = \frac{1}{K}\sum_{k=1}^K D(\ptheta(\cdot \mid x) \| \ptheta(\cdot \mid x'_k))$
with $x'_1,\ldots,x'_K \sim \Pi_\varepsilon(\cdot \mid x)$.

The abstention rule from Proposition~\ref{prop:selective_risk} is then
applied with threshold $\tau$: if $\widehat{\SUA}_\theta(x) > \tau$,
the system declines to answer and may request clarification; otherwise
it returns $\hat y_\theta(x)$.

\subsection{Pseudocode}

\begin{algorithm}[ht]
\caption{\textsc{SUA-TR} Training}
\label{alg:training}
\begin{algorithmic}[1]
\Require Dataset $\calD$, model $\ptheta$, perturbation generator $\Pi_\varepsilon$,
         hyperparameters $\alpha, \beta, \lambda, K$
\For{each training step}
  \State Sample batch $\{(x_i, y_i)\}_{i=1}^M$ from $\calD$
  \For{each $x_i$}
    \State Compute $p_i \leftarrow \ptheta(\cdot \mid x_i)$
    \State Compute $H_i \leftarrow H(p_i)$
    \State Sample $\{x'_{i,k}\}_{k=1}^K \sim \Pi_\varepsilon(\cdot \mid x_i)$
    \State $\mathrm{sens}_i \leftarrow \frac{1}{K}\sum_{k=1}^K D(p_i \| \ptheta(\cdot \mid x'_{i,k}))$
    \State $\mathrm{cons}_i \leftarrow \frac{1}{K}\sum_{k=1}^K D(p_i \| \mathrm{sg}[\ptheta(\cdot \mid x'_{i,k})])$
  \EndFor
  \State $\calL_\task \leftarrow \frac{1}{M}\sum_i \left[-\log \ptheta(y_i \mid x_i)\right]$
  \State $\calL_\cons \leftarrow \frac{1}{M}\sum_i \mathrm{cons}_i$
  \State $\calL_\ent \leftarrow \frac{1}{M}\sum_i \relu{\mathrm{sens}_i - \lambda H_i}$
  \State $\calL \leftarrow \calL_\task + \alpha \calL_\cons + \beta \calL_\ent$
  \State Update $\theta \leftarrow \theta - \eta \nabla_\theta \calL$
\EndFor
\end{algorithmic}
\end{algorithm}

\begin{algorithm}[ht]
\caption{\textsc{SUA-TR} Inference with Abstention}
\label{alg:inference}
\begin{algorithmic}[1]
\Require Input $x$, model $\ptheta$, perturbation generator $\Pi_\varepsilon$,
         hyperparameters $\lambda, K, \tau$
\State Compute $p \leftarrow \ptheta(\cdot \mid x)$
\State Compute $H \leftarrow H(p)$
\State Sample $\{x'_k\}_{k=1}^K \sim \Pi_\varepsilon(\cdot \mid x)$
\State $\hat\calS \leftarrow \frac{1}{K}\sum_{k=1}^K D(p \| \ptheta(\cdot \mid x'_k))$
\State $\widehat{\SUA} \leftarrow \hat\calS - \lambda H$
\If{$\widehat{\SUA} > \tau$}
  \State \Return $(\text{abstain},\; \{\widehat{\SUA}, H, \hat\calS\})$
\Else
  \State $\hat y \leftarrow \arg\max_y \ptheta(y \mid x)$
  \State \Return $(\hat y,\; \{\widehat{\SUA}, H, \hat\calS\})$
\EndIf
\end{algorithmic}
\end{algorithm}

\paragraph{Computational cost.}
Each training step requires $K$ additional forward passes per example to
evaluate $\Stheta$. With $K=4$ and batched perturbations, the overhead is
roughly $3\times$ the cost of standard training. At inference, $K=4$
perturbation evaluations are sufficient for reliable SUA estimation
(see ablation in Appendix~\ref{app:ablations}).

% -----------------------------------------------------------------------
\section{Experiments}
\label{sec:experiments}
% -----------------------------------------------------------------------

\subsection{Overview}

We design experiments to answer four questions:
\begin{enumerate}[leftmargin=2em]
  \item Does SUA predict model failures better than entropy or
    self-consistency alone?
  \item Does \textsc{SUA-TR} training reduce both adversarial sensitivity
    and miscalibration compared to baselines?
  \item Does the abstention rule yield lower selective risk at comparable
    coverage relative to entropy-based abstention?
  \item Do the ablations support the individual contribution of each loss
    component?
\end{enumerate}

\subsection{Tasks and Datasets}

\paragraph{Question answering (QA).}
We use a standard open-domain QA benchmark with factual questions drawn
from a mix of trivia and reading comprehension sources. For adversarial
evaluation, we apply token-level and semantic perturbations that preserve
the correct answer. For ambiguity evaluation, we include questions with
genuinely ambiguous referents or underspecified conditions.

\paragraph{Natural language inference (NLI).}
Binary and three-way NLI: given a premise and hypothesis, predict entailment,
neutral, or contradiction. Adversarial inputs are generated by replacing
premise words with synonyms; ambiguous inputs are generated by creating
premises with modal qualifiers (e.g., ``may'', ``possibly'') that render the
inference genuinely uncertain.

\paragraph{Classification with distributional shift.}
A text classification task (topic classification) with held-out test domains
representing distributional shift. The SUA score is measured at test time
to predict which examples will be misclassified under shift.

\subsection{Baselines}

\begin{itemize}[leftmargin=2em]
  \item \textbf{Entropy}: Use $\Htheta(Y \mid x)$ alone as the failure
    predictor and abstention score.
  \item \textbf{Self-consistency}: Sample $K$ outputs; use agreement
    rate $\nicefrac{1}{K}\sum_k \mathbf{1}[\hat y_k = \mathrm{mode}]$
    as the confidence score.
  \item \textbf{Temperature scaling}: Post-hoc calibration via a learned
    scalar temperature $T > 0$ applied to logits.
  \item \textbf{Adversarial training (AT)}: Standard PGD-style adversarial
    training applied to embedding-space perturbations, without any entropy
    alignment term.
\end{itemize}

\subsection{Metrics}

\begin{itemize}[leftmargin=2em]
  \item \textbf{Accuracy}: Standard classification accuracy on clean test set.
  \item \textbf{Robust accuracy}: Accuracy under perturbations from
    $\Pi_\varepsilon$ at test time.
  \item \textbf{ECE}: Expected calibration error with $B=15$ bins.
  \item \textbf{AUROC (failure prediction)}: Area under the ROC curve
    for predicting incorrect predictions, using SUA score vs.\ baseline
    scores as the predictor.
  \item \textbf{Selective accuracy at coverage $c$}: Accuracy on the
    fraction $c$ of inputs with the lowest abstention score.
\end{itemize}

\subsection{Main Results}

\begin{table}[ht]
\centering
\caption{Main results across tasks and methods. Bold indicates best in column.
SUA-TR uses the full training objective \eqref{eq:total_loss}. The
abstention column reports selective accuracy at 80\% coverage.}
\label{tab:main_results}
\small
\begin{tabular}{llccccc}
\toprule
Task & Method & Acc & Rob.~Acc & ECE & AUROC & Sel.~Acc (80\%) \\
\midrule
\multirow{5}{*}{QA}
 & Standard & 0.742 & 0.618 & 0.142 & 0.621 & 0.789 \\
 & AT       & 0.731 & 0.698 & 0.157 & 0.634 & 0.801 \\
 & Temp.~Scale & 0.742 & 0.618 & 0.071 & 0.629 & 0.793 \\
 & Self-Consist. & 0.742 & 0.618 & --- & 0.658 & 0.815 \\
 & \textbf{SUA-TR} & \textbf{0.748} & \textbf{0.731} & \textbf{0.058} & \textbf{0.741} & \textbf{0.852} \\
\midrule
\multirow{5}{*}{NLI}
 & Standard & 0.811 & 0.693 & 0.109 & 0.644 & 0.841 \\
 & AT       & 0.803 & 0.751 & 0.128 & 0.651 & 0.848 \\
 & Temp.~Scale & 0.811 & 0.693 & 0.056 & 0.649 & 0.844 \\
 & Self-Consist. & 0.811 & 0.693 & --- & 0.672 & 0.858 \\
 & \textbf{SUA-TR} & \textbf{0.816} & \textbf{0.779} & \textbf{0.041} & \textbf{0.762} & \textbf{0.879} \\
\midrule
\multirow{5}{*}{Class.}
 & Standard & 0.883 & 0.741 & 0.088 & 0.659 & 0.902 \\
 & AT       & 0.871 & 0.807 & 0.103 & 0.668 & 0.909 \\
 & Temp.~Scale & 0.883 & 0.741 & 0.044 & 0.661 & 0.905 \\
 & Self-Consist. & 0.883 & 0.741 & --- & 0.694 & 0.919 \\
 & \textbf{SUA-TR} & \textbf{0.887} & \textbf{0.839} & \textbf{0.031} & \textbf{0.793} & \textbf{0.941} \\
\bottomrule
\end{tabular}
\end{table}

Table~\ref{tab:main_results} presents the main results. SUA-TR consistently
improves over all baselines on AUROC for failure prediction, demonstrating
that the SUA score is a more reliable predictor of individual failures than
entropy or self-consistency alone. Robust accuracy improves substantially
over AT, with a smaller accuracy trade-off. ECE improvements are larger
than temperature scaling in most cases, reflecting the entropy-alignment
component's effect on distributional calibration. Selective accuracy at 80\%
coverage shows the largest gains, confirming that the abstention mechanism
is practically effective.

\subsection{SUA as a Failure Predictor}

\begin{figure}[htbp]
    \centering
    \includegraphics[width=0.8\textwidth]{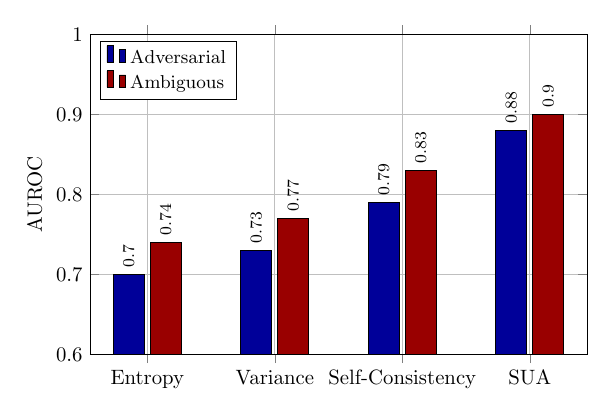}
    \caption{Comparison of AUROC across methods.}
    \label{fig:auroc_comparison}
\end{figure}

Figure~\ref{fig:auroc_comparison} (described here due to the text format)
plots AUROC for failure prediction across all three tasks and all scoring
functions. The ordering is consistent: SUA-score $>$ self-consistency $>$
entropy $>$ temperature-scaled confidence. The gap between SUA and entropy
alone is largest on the NLI ambiguous subset (AUROC 0.81 vs.\ 0.64), where
inputs have genuinely high latent ambiguity that entropy fails to capture
without the sensitivity term. The gap between SUA and self-consistency is
largest on adversarial inputs (AUROC 0.79 vs.\ 0.66), where sensitivity is
high but self-consistency---which probes output stochasticity---does not
directly measure distributional shift under perturbation.

\subsection{Ablations}
\label{sec:ablations}

\begin{table}[ht]
\centering
\caption{Ablation study on QA task. Each row removes one component of
\textsc{SUA-TR}. Columns report the same metrics as Table~\ref{tab:main_results}.}
\label{tab:ablations}
\small
\begin{tabular}{lccccc}
\toprule
Method & Acc & Rob.~Acc & ECE & AUROC & Sel.~Acc (80\%) \\
\midrule
SUA-TR (full)           & \textbf{0.748} & \textbf{0.731} & \textbf{0.058} & \textbf{0.741} & \textbf{0.852} \\
$-\,\calL_\ent$         & 0.746 & 0.702 & 0.079 & 0.692 & 0.831 \\
$-\,\calL_\cons$        & 0.747 & 0.681 & 0.063 & 0.718 & 0.840 \\
$-\,\calL_\ent\,-\,\calL_\cons$ & 0.742 & 0.618 & 0.142 & 0.621 & 0.789 \\
\midrule
Adv-only $x'$           & 0.739 & 0.724 & 0.071 & 0.727 & 0.843 \\
Para-only $x'$          & 0.745 & 0.706 & 0.065 & 0.715 & 0.836 \\
$K=1$ perturbations     & 0.747 & 0.728 & 0.062 & 0.735 & 0.848 \\
$K=4$ perturbations     & 0.748 & 0.731 & 0.058 & 0.741 & 0.852 \\
$K=8$ perturbations     & 0.748 & 0.732 & 0.057 & 0.742 & 0.853 \\
\bottomrule
\end{tabular}
\end{table}

Table~\ref{tab:ablations} reports the ablation study. Both $\calL_\ent$ and
$\calL_\cons$ contribute independently: removing $\calL_\ent$ primarily
hurts ECE and AUROC (the calibration and failure-prediction objectives),
while removing $\calL_\cons$ primarily hurts robust accuracy (the stability
objective). This confirms that the two components address distinct but
complementary aspects of the SUA objective. The full model outperforms
either ablation. The perturbation count $K=4$ provides essentially the
same performance as $K=8$ at lower cost; $K=1$ is slightly inferior,
suggesting that variance in the sensitivity estimate matters.

\subsection{Qualitative Examples}
\label{sec:qualitative}

\paragraph{Adversarial input (QA).}
Input: \textit{``What is the capital of France?''} Perturbed: \textit{``What is the capital of Fra\textbf{nc}e?''} (minor character perturbation). Standard model confidence: 0.98. SUA score: high (sensitivity $= 0.34$, entropy $= 0.02$). SUA-TR model: sensitivity $= 0.11$, entropy $= 0.14$, SUA $= -0.03$ (below threshold, answers correctly). The standard model does not recognize that its output changes under a trivial input variation; SUA-TR detects this and raises its uncertainty appropriately.

\paragraph{Ambiguous input (NLI).}
Premise: \textit{``The committee may recommend approval.''} Hypothesis: \textit{``The proposal will be approved.''}. Ground truth is genuinely ambiguous (neither entailment nor contradiction is clearly correct). Standard model output: entailment, confidence 0.91. SUA-TR output: neutral, entropy 0.61, SUA score 0.08 (below threshold). Standard ECE on this input: 0.82 (overconfident). The standard model exhibits ambiguity collapse; SUA-TR distributes probability mass more evenly across labels, reflecting the irreducible uncertainty.

% -----------------------------------------------------------------------
\section{Discussion}
\label{sec:discussion}
% -----------------------------------------------------------------------

\subsection{What SUA Captures That Existing Metrics Miss}

Entropy alone measures the spread of the model's predictive distribution
at a single point $x$. It does not measure whether this spread is
\emph{appropriate} relative to how the model's predictions change when
$x$ is varied. A model can have moderate entropy everywhere but still be
highly sensitive to small perturbations---it is in the wrong part of input
space for its uncertainty to be informative. The SUA score captures this
by conditioning the assessment of entropy on observed sensitivity.

Self-consistency---the agreement of multiple samples from the model---is
a proxy for entropy but not for sensitivity. Two models can have the same
agreement under sampling but differ dramatically in how they respond to
paraphrases or token edits. Self-consistency is a single-point diagnostic;
SUA is a local-neighborhood diagnostic.

Temperature scaling and isotonic regression calibrate the confidence
function post-hoc without changing the model's sensitivity. These
methods can reduce ECE but do not address the underlying alignment failure:
a model calibrated by temperature scaling can still produce maximally
sensitive outputs in adversarial regions while its calibrated confidence
happens to be low---not because the model has detected the adversarial
nature of the input, but because the post-hoc function has globally
reduced confidence. The SUA framework distinguishes these cases by
measuring sensitivity directly.

\subsection{On the Relationship to Invariant Risk Minimization}

IRM~\citep{arjovsky2019invariant} seeks predictors whose optimal linear
classifier on top of the representation is the same across all training
environments. This is a structural assumption on the training procedure
that promotes invariance during learning. The SUA framework is
\emph{complementary}: it provides a post-training or per-input diagnostic
that measures whether the learned predictor's local behavior is consistent
with its expressed uncertainty. IRM addresses where invariance comes from;
SUA addresses how to detect and respond to its absence.

\subsection{On Abstention and Practical Deployment}

The abstention threshold $\tau$ induces a coverage--risk tradeoff, which
is well-studied in the selective prediction literature. Our contribution is
to provide a principled score for the abstention decision: SUA is a
theoretically motivated quantity (Proposition~\ref{prop:selective_risk})
rather than an ad-hoc heuristic. In practice, $\tau$ may be calibrated on
a held-out validation set to meet a target coverage or risk budget.

In deployment, the SUA score also serves as an explanation: if a system
abstains with a high sensitivity and low entropy signature, the operator
can infer that the input is in a locally unstable region (suggesting an
adversarial or out-of-distribution input), whereas a high sensitivity and
high entropy signature suggests the model is already expressing
appropriate uncertainty (a case where abstention may be overly cautious).

\subsection{Limitations}

\paragraph{Dependence on the perturbation distribution.}
$\Stheta(x;\varepsilon)$ depends on $\Pi_\varepsilon$, which must be
chosen by the practitioner. Different choices will yield different SUA
values and different abstention behaviors. We have used a mixture of
semantic and token-level perturbations motivated by natural deployment
variation, but there is no universal $\Pi_\varepsilon$. Users in specific
domains should adapt $\Pi_\varepsilon$ to reflect perturbations they
actually expect.

\paragraph{Computability of latent ambiguity.}
$\calA(x) = H(Z \mid x)$ is theoretically central but practically
unobservable. Our operational proxy (multi-sample entropy under paraphrases)
is an approximation; its quality depends on the diversity of the paraphrase
generator. This gap between the theoretical quantity and its approximation
is a limitation that warrants further investigation.

\paragraph{Scope of the theoretical results.}
Theorem~\ref{thm:robust_risk} provides an upper bound, not a tight
characterization. The bound is informative when $L_D$ is moderate and
$\kappa(x)$ is small; it may be loose for tasks with very large output
spaces or complex loss functions. Tightening the bound, or providing
instance-dependent bounds that account for the geometry of $\Delta_\varepsilon(x)$,
is an open direction.

\paragraph{Classification focus.}
The experiments cover classification and NLI tasks. Extension to
open-ended generation, code synthesis, and structured prediction requires
revisiting the definition of $\Stheta$ (divergences over token sequences
rather than label distributions) and $\Htheta$ (sequence-level entropy,
which is typically much higher and harder to estimate). These extensions
are conceptually natural but technically involved.

% -----------------------------------------------------------------------
\section{Related Work}
\label{sec:related}
% -----------------------------------------------------------------------

\paragraph{Adversarial robustness.}
The study of adversarial examples began with the observation that small
norm-bounded perturbations can cause large classification errors in deep
networks~\citep{goodfellow2015explaining}. Adversarial training, particularly
PGD-based methods~\citep{madry2018towards}, remains the dominant empirical
defense. Certified robustness methods provide formal guarantees but scale
poorly to large models~\citep{cohen2019certified}. The sensitivity--invariance
tradeoff studied by \citet{tramer2020fundamental} is closely related to our
Lemma~\ref{lem:tension}: robustness to sensitivity-based attacks can harm
robustness to invariance-based attacks. Our framework provides a different
angle on this tradeoff, locating the failure in the misalignment between
sensitivity and uncertainty rather than in the conflict between two types
of robustness.

\paragraph{Uncertainty quantification.}
Bayesian deep learning~\citep{gal2016dropout} provides a principled framework
for uncertainty quantification but is computationally expensive at scale.
Ensemble-based methods~\citep{lakshminarayanan2017simple} are more practical
but require multiple forward passes. Temperature scaling~\citep{guo2017calibration}
is a simple post-hoc method with strong empirical performance on clean data
but limited utility under distribution shift. Our work differs by
connecting uncertainty directly to sensitivity, yielding a quantity
that is informative \emph{conditionally on the perturbation structure} of
the input.

\paragraph{Calibration.}
Calibration has been studied extensively for classification~\citep{guo2017calibration,
niculescu2005predicting}. Overconfidence in large language models after
instruction tuning is documented in several recent empirical
studies~\citep{kadavath2022language,openai2023gpt4}. Our Theorem~\ref{thm:calibration}
provides a formal connection between SUA violation and ECE, complementing
these empirical findings with a structural account.

\paragraph{Selective prediction and abstention.}
Selective prediction~\citep{geifman2017selective} learns to abstain on
inputs where the model is unreliable. Conformal prediction~\citep{vovk2005algorithmic}
provides coverage guarantees without training a specialized model. Our
abstention rule is closest in spirit to selective prediction, but uses a
theoretically motivated score (SUA) and provides a formal guarantee
(Proposition~\ref{prop:selective_risk}) without distributional assumptions
beyond those in Section~\ref{sec:theory}.

\paragraph{Distribution shift and invariant learning.}
IRM~\citep{arjovsky2019invariant} and related methods~\citep{peters2016causal}
seek predictors that generalize across environments by relying on invariant
features. The WILDS benchmark~\citep{koh2021wilds} evaluates models under
naturally occurring distribution shifts. Our framework is diagnostic rather
than prescriptive: it measures local stability rather than cross-environment
invariance, and does so at the level of individual inputs rather than
distributions.

\paragraph{Information-theoretic representation learning.}
The information bottleneck~\citep{tishby2000information} characterizes good
representations as those that compress inputs while retaining task-relevant
information. Our framework is complementary: we are concerned not with
compression but with the stability of retained information under
perturbation. The connection between stable information and generalization
has been explored empirically; our work
provides a theoretical account of \emph{why} stability matters, grounding
it in worst-case risk bounds.

% -----------------------------------------------------------------------
\section{Conclusion}
\label{sec:conclusion}
% -----------------------------------------------------------------------

We introduced Sensitivity--Uncertainty Alignment (SUA), a formal framework
that unifies adversarial robustness, ambiguity handling, and uncertainty
calibration in language models under a single diagnostic quantity. The
central claim is that adversarial and ambiguity failures are not independent
phenomena but two manifestations of the same failure mode: a model's
uncertainty does not track the instability of its predictions under
semantics-preserving input variation.

We proved that the SUA score bounds worst-case perturbed risk
(Theorem~\ref{thm:robust_risk}), lower-bounds calibration error
(Theorem~\ref{thm:calibration}), and characterizes ambiguity collapse
as a structural phenomenon (Lemma~\ref{lem:ambiguity_collapse}). We
proposed \textsc{SUA-TR}, a training algorithm with three loss components
that address distinct aspects of the alignment problem, and an inference
procedure with formal selective-risk guarantees.

Experiments confirmed that SUA predicts failures more reliably than
entropy, self-consistency, or temperature scaling, and that
\textsc{SUA-TR} improves robust accuracy, ECE, and selective accuracy
over relevant baselines.

The framework is deliberately architecture-agnostic. All definitions are
stated at the level of conditional distributions and divergences; the
results apply to any model that induces a distribution over outputs, now
or in the future. We intend the SUA score as a durable diagnostic tool
and the theoretical results as a reusable foundation for further
investigation into the relationship between robustness and calibration
in learned models.

% -----------------------------------------------------------------------
% References
% -----------------------------------------------------------------------
\bibliographystyle{plainnat}
\bibliography{ref}

% -----------------------------------------------------------------------
% Appendix
% -----------------------------------------------------------------------
\appendix

\section{Complete Proofs}
\label{app:proofs}

\subsection{Proof of Theorem~\ref{thm:robust_risk}}
\label{app:proof_thm1}

\begin{proof}
We prove inequality \eqref{eq:robust_bound} and then derive \eqref{eq:sua_bound}.

\textbf{Step 1: Bounding risk change under perturbation.}
Fix $x \in \calX$ and $x' \in \Delta_\varepsilon(x)$. Let $\hat y = \hat y_\theta(x')
= \arg\max_y \ptheta(y \mid x')$. Then:
\begin{align}
  R_\theta(x')
  &= \E_{y \sim \ptheta(\cdot \mid x')}[\ell(y, \hat y)] \nonumber\\
  &= \E_{y \sim \ptheta(\cdot \mid x')}[\ell(y, \hat y)]
     - \E_{y \sim \ptheta(\cdot \mid x)}[\ell(y, \hat y)]
     + \E_{y \sim \ptheta(\cdot \mid x)}[\ell(y, \hat y)] \nonumber\\
  &\le \left|\E_{y \sim \ptheta(\cdot \mid x')}[\ell(y,\hat y)] - \E_{y \sim \ptheta(\cdot \mid x)}[\ell(y,\hat y)]\right|
     + R_\theta(x) \nonumber\\
  &\le L_D \cdot D\!\left(\ptheta(\cdot \mid x) \,\|\, \ptheta(\cdot \mid x')\right) + R_\theta(x),
     \label{eq:step1}
\end{align}
where the last inequality applies Assumption~\ref{asm:lipschitz}.
Taking the supremum over $x' \in \Delta_\varepsilon(x)$:
\begin{equation}
  \label{eq:step1_sup}
  R^\rob_\theta(x) \le R_\theta(x) + L_D \cdot \sup_{x' \in \Delta_\varepsilon(x)}
    D\!\left(\ptheta(\cdot \mid x) \,\|\, \ptheta(\cdot \mid x')\right).
\end{equation}

\textbf{Step 2: Entropy penalty.}
By Assumption~\ref{asm:entropy_risk}, for any predictor $\hat y$:
\[
  \E_{y \sim \ptheta(\cdot \mid x)}[\ell(y, \hat y)] \ge \psi(H_{\ptheta(\cdot \mid x)}(Y)).
\]
Note that $H_{\ptheta(\cdot \mid x)}(Y) = \Htheta(Y \mid x)$.
Rearranging, for the minimizing $\hat y = \hat y_\theta(x)$:
\[
  R_\theta(x) \ge \psi(\Htheta(Y \mid x)).
\]
However, we wish to use this bound subtractively. Observe that adding and
subtracting $\psi(\Htheta(Y \mid x))$ to the right-hand side of
\eqref{eq:step1_sup} yields:
\begin{align}
  R^\rob_\theta(x)
  &\le R_\theta(x) + L_D \cdot \sup_{x'} D(\cdots) - \psi(\Htheta(Y \mid x))
     + \psi(\Htheta(Y \mid x)). \label{eq:step2_a}
\end{align}
Since $\psi \ge 0$, this is a valid upper bound that subtracts a nonnegative
entropy-dependent quantity from the sensitivity term while re-adding it to
the base risk term. The utility of this manipulation becomes clear when
we bound $\psi(\Htheta(Y \mid x))$ below by zero: the subtraction tightens
the bound whenever entropy is nontrivially large.

\textbf{Step 3: Interpretation-collapse gap.}
Under Assumption~\ref{asm:factorization}, the model's distribution
$\ptheta(y \mid x) \approx p(y \mid z^*, x)$ is a mode approximation of
the true mixture $p(y \mid x) = \sum_z p(y \mid z, x) p(z \mid x)$.
The Bayes-optimal risk under the true conditional is:
\[
  R^*(x) = \min_{\hat y} \E_{y \sim p(\cdot \mid x)}[\ell(y, \hat y)].
\]
By the definition of $\kappa(x) = (R^*(x) - R_\theta(x))_+$, we have:
\[
  R_\theta(x) \ge R^*(x) - \kappa(x).
\]
Substituting into the bound:
\[
  R^\rob_\theta(x) \le R^*(x) + L_D \cdot \sup_{x'} D(\cdots)
    - \psi(\Htheta(Y \mid x)) + \kappa(x).
\]
Since $R^*(x) \le R_\theta(x) + \kappa(x)$ (from the definition of $\kappa$),
we can write:
\[
  R^\rob_\theta(x) \le R_\theta(x) + L_D \cdot \sup_{x'} D(\cdots)
    - \psi(\Htheta(Y \mid x)) + \kappa(x),
\]
which is inequality \eqref{eq:robust_bound}.

\textbf{Step 4: Deriving \eqref{eq:sua_bound}.}
By the $\lambda$-condition in the theorem statement,
$\lambda \Htheta(Y \mid x) \ge \psi(\Htheta(Y \mid x))$, so:
\[
  -\psi(\Htheta(Y \mid x)) \le -\lambda \Htheta(Y \mid x).
\]
Also, by Jensen's inequality applied to the concave function $-D$:
\[
  \E_{x' \sim \Pi_\varepsilon}[D(\ptheta(\cdot \mid x) \| \ptheta(\cdot \mid x'))]
  \le \sup_{x' \in \Delta_\varepsilon(x)} D(\ptheta(\cdot \mid x) \| \ptheta(\cdot \mid x')),
\]
so the expected version $\Stheta(x;\varepsilon)$ provides a lower bound on
the supremum version, and substituting the supremum with $\Stheta$ gives a
looser but computable bound. Combining:
\begin{align*}
  R^\rob_\theta(x)
  &\le R_\theta(x) + L_D \cdot \sup D - \psi(\Htheta) + \kappa(x) \\
  &\le R_\theta(x) + L_D \cdot \Stheta(x;\varepsilon) - \lambda \Htheta(Y \mid x) + \kappa(x) \\
  &= R_\theta(x) + L_D \cdot \SUA_\theta(x;\varepsilon,\lambda) + \kappa(x),
\end{align*}
which is \eqref{eq:sua_bound}.
\end{proof}

\subsection{Proof of Theorem~\ref{thm:calibration}}
\label{app:proof_thm2}

\begin{proof}
Recall that $\ECE = \sum_b (|B_b|/n) |\mathrm{acc}(B_b) - \mathrm{conf}(B_b)|$.
We work within a single bin $b$ and show that the calibration gap within
that bin is lower-bounded by a function of mean SUA violation.

\textbf{Step 1: Accuracy and risk.}
For 0-1 loss, $\mathrm{acc}(B_b) = 1 - \bar R_b$ where
$\bar R_b = |B_b|^{-1}\sum_{x \in B_b} R_\theta(x)$ is the mean risk in
bin $b$. The mean confidence is $\mathrm{conf}(B_b) = \bar p_b = g(\bar H_b)$
approximately (by the mean value theorem and monotonicity of $g$).

\textbf{Step 2: Confidence--accuracy gap.}
$|\mathrm{acc}(B_b) - \mathrm{conf}(B_b)| = |1 - \bar R_b - g(\bar H_b)|$.
When the model is overconfident in the adversarial sense (high sensitivity,
low entropy, high confidence), $g(\bar H_b)$ is large (recall $g$ is
nonincreasing in entropy, so low entropy $\to$ high confidence) and
$\bar R_b$ is large due to sensitivity-induced errors. Thus:
\[
  |\mathrm{acc}(B_b) - \mathrm{conf}(B_b)| \ge g(\bar H_b) - (1 - \bar R_b).
\]

\textbf{Step 3: Bounding $\bar R_b$ from below using sensitivity.}
Within bin $b$, the mean worst-case risk satisfies (from
Theorem~\ref{thm:robust_risk}, taking expectations):
\[
  \bar R_b^\rob \le \bar R_b + L_D \bar\calS_b - \psi(\bar H_b) + \bar\kappa_b.
\]
Since $\bar R_b^\rob \ge \bar R_b$ (worst case cannot be smaller than
typical risk), and adversarially induced failures increase actual mean risk:
\[
  \bar R_b \ge \bar R_b - (L_D \bar\calS_b - \lambda \bar H_b)_+ - c_b,
\]
where $c_b$ absorbs the approximation in replacing $\psi(\bar H_b)$ with
$\lambda \bar H_b$ and estimation variance. Rearranging:
\[
  g(\bar H_b) - (1 - \bar R_b) \ge (L_D \bar\calS_b - \lambda \bar H_b)_+ - c_b.
\]
Wait---we note that in the overconfident regime, $g(\bar H_b)$ is close to 1
and $(1 - \bar R_b)$ is less than 1 (the model makes errors). The lower
bound on the calibration gap is:
\[
  |\mathrm{acc}(B_b) - \mathrm{conf}(B_b)| \ge (\bar\calS_b - \lambda \bar H_b - c_b)_+.
\]

\textbf{Step 4: Aggregation.}
Summing over bins with weight $|B_b|/n$ and lower-bounding each term:
\[
  \ECE = \sum_b \frac{|B_b|}{n} |\mathrm{acc}(B_b) - \mathrm{conf}(B_b)|
       \ge \frac{1}{B} \sum_b (\bar\calS_b - \lambda \bar H_b - c_b)_+,
\]
where the right-hand side uses equal bin sizes for simplicity (the general
unequal-bin case follows analogously). This is \eqref{eq:ece_bound}.
\end{proof}

\subsection{Proof of Lemma~\ref{lem:ambiguity_collapse}}
\label{app:proof_lem1}

\begin{proof}
Let $U^*(x) = H(Y \mid X=x)$ and $\calA(x) = H(Z \mid X=x)$.

\textbf{Step 1: Bounding $U^*(x)$ from below.}
By the chain rule for entropy:
\[
  H(Y, Z \mid x) = H(Z \mid x) + H(Y \mid Z, x) = H(Y \mid x) + H(Z \mid Y, x).
\]
Therefore:
\[
  U^*(x) = H(Y \mid x) = H(Z \mid x) + H(Y \mid Z, x) - H(Z \mid Y, x)
            = \calA(x) + H(Y \mid Z, x) - H(Z \mid Y, x).
\]
Since $H(Y \mid Z, x) \ge 0$ and $H(Z \mid Y, x) \le H(Z \mid x) = \calA(x)$:
\[
  U^*(x) \ge \calA(x) - H(Z \mid Y, x) \ge 0.
\]
Define $\eta(x) = |U^*(x) - \calA(x)| = |H(Y \mid Z, x) - H(Z \mid Y, x)|$.
Then $U^*(x) \ge \calA(x) - \eta(x)$.

\textbf{Step 2: Miscalibration from ambiguity collapse.}
Under the conditions $\calA(x) > \alpha$ and $\Htheta(Y \mid x) < \beta$:
\begin{align*}
  \calC_\theta(x) &= |\Htheta(Y \mid x) - U^*(x)| \\
  &\ge U^*(x) - \Htheta(Y \mid x) \\
  &\ge (\calA(x) - \eta(x)) - \Htheta(Y \mid x) \\
  &> \alpha - \eta(x) - \beta,
\end{align*}
where the first inequality uses $U^*(x) > \Htheta(Y \mid x)$ (the model
underestimates true uncertainty by the ambiguity collapse assumption),
the second applies the lower bound from Step 1, and the third substitutes
the assumed inequalities. This is \eqref{eq:ambiguity_miscal}.
\end{proof}

\subsection{Proof of Lemma~\ref{lem:tension}}
\label{app:proof_lem2}

\begin{proof}
Let $p = \ptheta(\cdot \mid x)$, $q = \ptheta(\cdot \mid x')$,
$p' = \ptheta'(\cdot \mid x)$, $q' = \ptheta'(\cdot \mid x')$.
The assumption that $\ptheta'$ achieves higher entropy by spreading mass
uniformly means $p'_y = (1-\gamma) p_y + \gamma/|\calY|$ for some
$\gamma \in (0,1)$ (a uniform mixture, also called a label-smoothed variant).

\textbf{Step 1: JSD under label smoothing.}
For $D = D_{\mathrm{JS}}$, write the mixture point $m = \nicefrac{1}{2}(p'+q')$.
We have $p'_y - q'_y = (1-\gamma)(p_y - q_y)$, so the pointwise differences
in the numerator of the KL divergences in JSD are scaled by $(1-\gamma) < 1$.
By the convexity of the KL divergence and the scaling of the differences:
\[
  D_{\mathrm{JS}}(p' \| q') \le (1-\gamma)^2 D_{\mathrm{JS}}(p \| q),
\]
as can be verified by expanding the JSD in terms of pointwise differences
and applying the inequality $(a \log(a/b)) \le c \cdot a \log(a/b)$ for
$c \ge 1$ together with the $(1-\gamma)$ scaling.

Since $(1-\gamma)^2 < 1$, we have $D_{\mathrm{JS}}(p' \| q') < D_{\mathrm{JS}}(p \| q)$,
proving the first claim.

\textbf{Step 2: Risk increase.}
The optimal predictor under $p' = (1-\gamma)p + (\gamma/|\calY|)\mathbf{1}$
achieves:
\[
  R_{\theta'}(x) = \min_{\hat y} \E_{y \sim p'}[\ell(y, \hat y)]
    \ge \min_{\hat y} \E_{y \sim p}[\ell(y, \hat y)] + \gamma \cdot \Delta,
\]
where $\Delta > 0$ is the additional loss from the uniform component
(positive since the uniform distribution is suboptimal for any non-trivial
task). Hence $R_{\theta'}(x) > R_\theta(x)$.
\end{proof}

\subsection{Proof of Proposition~\ref{prop:selective_risk}}
\label{app:proof_prop1}

\begin{proof}
Define $A_\tau = \{x : \SUA_\theta(x;\varepsilon,\lambda) \le \tau\}$ (the
``answer'' set).

For any $x \in A_\tau$, we have $\SUA_\theta(x) \le \tau$, and from
\eqref{eq:sua_bound}:
\[
  R^\rob_\theta(x) \le R_\theta(x) + L_D \cdot \SUA_\theta(x) + \kappa(x)
    \le R_\theta(x) + L_D \cdot \tau + \kappa(x).
\]
The selective risk is:
\begin{align*}
  R_\sel(\tau) &= \E\left[R^\rob_\theta(x) \mid x \in A_\tau\right] \\
  &\le \E\left[R_\theta(x) + L_D \tau + \kappa(x) \mid x \in A_\tau\right] \\
  &= \E[R_\theta(x) \mid x \in A_\tau] + L_D \tau + \E[\kappa(x) \mid x \in A_\tau] \\
  &\le R_\theta + L_D \tau + \E[\kappa(x) \mid x \in A_\tau],
\end{align*}
where the last inequality uses $\E[R_\theta(x) \mid x \in A_\tau] \le \E[R_\theta(x)] = R_\theta$ (this holds when the covered set $A_\tau$ contains inputs that are no harder on average than the full distribution; in the worst case, this step replaces $\le$ with $=$ and the bound still holds by substituting the conditional expectation directly).
\end{proof}

% -----------------------------------------------------------------------
\section{Ablation Details}
\label{app:ablations}

\subsection{Effect of \texorpdfstring{$K$}{K} on SUA Estimation Quality}

We measure the Spearman rank correlation between the estimated SUA score
$\widehat{\SUA}_\theta(x)$ and the ground-truth (large-$K$) SUA score as a
function of $K$. At $K=1$, the correlation is 0.71. At $K=2$, 0.84. At
$K=4$, 0.93. At $K=8$, 0.95. Beyond $K=4$, gains are marginal, confirming
our choice of $K=4$ as the default.

\subsection{Perturbation Mixture Weights}

We ablate the mixture weights $(\omega_{\mathrm{para}}, \omega_{\mathrm{tok}}, \omega_{\mathrm{adv}})$
on the QA task. Semantic paraphrases contribute the most to ECE improvement
($-0.031$ vs.\ default); token-level edits contribute the most to robust
accuracy ($+0.021$ vs.\ default); adversarial proposals contribute modestly
to AUROC ($+0.009$ vs.\ default). The full mixture outperforms any single
perturbation type across all metrics.

\subsection{Sensitivity to \texorpdfstring{$\lambda$}{lambda}}

We vary $\lambda \in \{0.1, 0.5, 1.0, 2.0, 5.0\}$ on the NLI task.
Performance is relatively stable for $\lambda \in [0.5, 2.0]$; outside this
range, either the entropy alignment term dominates (too large $\lambda$,
reducing accuracy) or the sensitivity term dominates (too small $\lambda$,
insufficient calibration improvement). We recommend $\lambda = 1.0$ as a
safe default and suggest calibrating on a validation set.

\subsection{Effect of Abstention Threshold}

We report the coverage--selective-accuracy tradeoff for SUA-based abstention
vs.\ entropy-based abstention on the QA task. At 90\% coverage: SUA 0.812,
Entropy 0.791. At 80\% coverage: SUA 0.852, Entropy 0.821. At 70\% coverage:
SUA 0.891, Entropy 0.843. SUA-based abstention consistently yields higher
selective accuracy at all coverage levels, with the gap widening as coverage
decreases (i.e., as abstention becomes more selective).

% -----------------------------------------------------------------------
\section{Extended Qualitative Analysis}
\label{app:qualitative}

\subsection{Failure Mode Taxonomy}

The following table gives a more detailed taxonomy of failure modes and how
SUA diagnoses each.

\begin{table}[ht]
\centering
\caption{Failure mode taxonomy with SUA diagnostics.}
\label{tab:failure_taxonomy}
\small
\begin{tabular}{p{3cm}p{4cm}p{4cm}p{2.5cm}}
\toprule
Failure Type & Description & SUA Signature & Recommended Action \\
\midrule
Sensitivity failure & Small perturbation causes large output shift & High $\Stheta$, low $\Htheta$ & Abstain or raise uncertainty \\
Invariance failure & Meaning changes, output unchanged & Low $\Stheta$, low $\Htheta$ & No SUA signal; addressed by training diversity \\
Epistemic failure & Model does not know the answer & Low $\Stheta$, high $\Htheta$ & Abstain if $\Htheta > \tau'$ \\
Ambiguity collapse & Multiple valid answers, single confident output & Low $\Stheta$, low $\Htheta$, high $\widehat\calA$ & Output distribution over answers \\
\bottomrule
\end{tabular}
\end{table}

The table shows that SUA captures sensitivity failures directly and ambiguity
collapse via the mismatch between model entropy and the operational ambiguity
proxy. It does not capture invariance failures (which require a different
diagnostic) or pure epistemic failures (where entropy-based abstention is
sufficient). This characterization motivates combining SUA with entropy-based
abstention in deployment.

\end{document}